\newif\ifisarxiv
\newif\ifislipics
\newif\ifislncs
\isarxivtrue
\ifisarxiv
\documentclass{article}
\usepackage{arxiv}
\fi
\ifislipics
\documentclass[letterpaper,USenglish,cleveref, autoref, thm-restate]{lipics-v2021}
\fi
\ifislncs
\documentclass[runningheads]{llncs}
\fi

\usepackage[utf8]{inputenc}
\usepackage{multirow}
\usepackage{amsmath}
\usepackage{amssymb}
\usepackage{multicol}
\usepackage{appendix}
\usepackage{array}
\usepackage{makecell}
\usepackage[table]{xcolor}
\usepackage{enumerate}
\usepackage{makecell}
\usepackage{listings}

\DeclareMathOperator*{\argmax}{arg\,max}
\newcommand{\Var}{\textnormal{Var}}
\newcommand{\PPAD}{\textup{PPAD}}
\newcommand{\PLS}{\textup{PLS}}
\newcommand{\CLS}{\textup{CLS}}
\newcommand{\RP}{\textup{RP}}
\newcommand{\Ptime}{\textup{P}}
\newcommand{\veca}{\textbf{\textup{a}}}
\newcommand{\vecu}{\textbf{\textup{u}}}

\newcommand{\melPLPG}{(m,\epsilon,\lambda)\textnormal{-\sc{PureLPG}}}
\newcommand{\melLPG}{(m,\epsilon,\lambda)\textnormal{-\sc{MixedLPG}}}
\newcommand{\PLPG}[3]{\ensuremath{(#1,#2,#3)\textnormal{-\sc{PureLPG}}}}
\newcommand{\LPG}[3]{\ensuremath{(#1,#2,#3)\textnormal{-\sc{MixedLPG}}}}

\newcommand{\E}{\mathbb{E}}
\newcommand{\R}{\mathbb{R}}

\newcommand{\N}{\mathbb{N}}

\newcommand{\EotL}{\textnormal{\sc{End of Line}}}

\newcommand{\PP}{\textup{P}}
\newcommand{\NP}{\textup{NP}}

\newcommand{\coNP}{\textup{co-NP}}
\newcommand{\vecp}{\textbf{\textup{p}}}

\newcommand{\reg}{\textnormal{reg}}

\newcommand{\poly}{\textnormal{\textup{poly}}}

\newcommand{\vece}{\textbf{\textup{e}}}
\newcommand{\veco}{\boldsymbol{1}}

\newcommand{\vecc}{\textbf{\textup{c}}}
\newcommand{\matL}{\textbf{\textup{L}}}
\newcommand{\vecb}{\textbf{\textup{b}}}

\newcommand{\BPP}{\textup{BPP}}

\newcommand{\br}{\textnormal{br}}

\pdfoutput=1 %uncomment to ensure pdflatex processing (mandatatory e.g. to submit to arXiv)
\ifisarxiv
\fi
\ifislipics
\fi
\ifislncs
\fi

\title{PPAD-Complete Pure Approximate Nash Equilibria in Lipschitz Games}

\ifislipics

\author{Paul W. Goldberg}{Department of Computer Science, University of Oxford, United Kingdom\and \url{http://www.cs.ox.ac.uk/people/paul.goldberg/index1.html} }{paul.goldberg@cs.ox.ac.uk}{https://orcid.org/
0000-0002-5436-7890}{}
\author{Matthew J. Katzman}{Department of Computer Science, University of Oxford, United Kingdom\and \url{https://www.cs.ox.ac.uk/people/matthew.katzman/} }{matthew.katzman@cs.ox.ac.uk}{https://orcid.org/
0000-0001-8147-9110}{Supported by an Oxford-DeepMind Studentship}%TODO mandatory, please use full name; only 1 author per \author macro; first two parameters are mandatory, other parameters can be empty. Please provide at least the name of the affiliation and the country. The full address is optional. Use additional curly braces to indicate the correct name splitting when the last name consists of multiple name parts.
\authorrunning{P.\,W. Goldberg and M.\,J. Katzman} %TODO mandatory. First: Use abbreviated first/middle names. Second (only in severe cases): Use first author plus 'et al.'
\Copyright{Paul W. Goldberg and Matthew J. Katzman} %TODO mandatory, please use full first names. LIPIcs license is "CC-BY";  http://creativecommons.org/licenses/by/3.0/

 %TODO mandatory: Please choose ACM 2012 classifications from https://dl.acm.org/ccs/ccs_flat.cfm 

 %TODO mandatory; please add comma-separated list of keywords

\category{} %optional, e.g. invited paper

\relatedversion{} %optional, e.g. full version hosted on arXiv, HAL, or other respository/website
\EventEditors{John Q. Open and Joan R. Access}
\EventNoEds{2}
\EventLongTitle{42nd Conference on Very Important Topics (CVIT 2016)}
\EventShortTitle{CVIT 2016}
\EventAcronym{CVIT}
\EventYear{2016}
\EventDate{December 24--27, 2016}
\EventLocation{Little Whinging, United Kingdom}
\EventLogo{}
\SeriesVolume{42}
\ArticleNo{23}
\nolinenumbers %uncomment to disable line numbering
\fi
\ifisarxiv
\usepackage{hyperref}
\author{
 Paul W. Goldberg \href{https://orcid.org/0000-0002-5436-7890}{\orcidsymbol} \\
  Department of Computer Science\\
  University of Oxford\\
  Oxford OX1 3QD, UK \\
  \texttt{paul.goldberg@cs.ox.ac.uk} \\
   \And
 Matthew Katzman \href{https://orcid.org/0000-0001-8147-9110}{\orcidsymbol}\thanks{This author was supported by an Oxford-DeepMind Studentship during the development of this work.} \\
  Department of Computer Science\\
  University of Oxford\\
  Oxford OX1 3QD, UK \\
  \texttt{matthew.katzman@cs.ox.ac.uk} \\
}
\date{\today}
\usepackage{amsthm}
\newtheorem{theorem}{Theorem}[section]
\newtheorem{corollary}[theorem]{Corollary}
\newtheorem{lemma}[theorem]{Lemma}
\theoremstyle{definition}
\newtheorem{definition}{Definition}
\theoremstyle{remark}
\newtheorem{remark}{Remark}
\fi
\ifislncs
\author{Paul W. Goldberg\inst{1}\orcidID{0000-0002-5436-7890} \and\\
Matthew Katzman\inst{1}\orcidID{0000-0001-8147-9110}}
\authorrunning{P. W. Goldberg and M. Katzman}
\institute{Department of Computer Science, University of Oxford, Oxford, England\\
\email{\{paul.goldberg,matthew.katzman\}@cs.ox.ac.uk}}
\fi
\begin{document}

\maketitle

\begin{abstract}
Lipschitz games, in which there is a limit $\lambda$ (the Lipschitz value of the game) on how much a player's payoffs may change when some other player deviates, were introduced about 10 years ago by Azrieli and Shmaya.  They showed via the probabilistic method that $n$-player Lipschitz games with $m$ strategies per player have {\em pure} $\epsilon$-approximate Nash equilibria, for $\epsilon \geq\lambda\sqrt{8n \log (2mn)}$.
Here we provide the first hardness result for the corresponding computational problem, showing that even for a simple class of Lipschitz games (Lipschitz polymatrix games), finding pure $\epsilon$-approximate equilibria is $\PPAD$-complete, for suitable pairs of values $(\epsilon(n), \lambda(n))$.  Novel features of this result include \emph{both} the proof of $\PPAD$ hardness (in which we apply a population game reduction from unrestricted polymatrix games) and the proof of containment in $\PPAD$ (by derandomizing the selection of a pure equilibrium from a mixed one).  In fact, our approach implies containment in $\PPAD$ for any class of Lipschitz games where payoffs from mixed-strategy profiles can be deterministically computed.

% \keywords{Equilibrium computation \and Lipschitz games \and Population games \and PPAD}
\end{abstract}

\section{Introduction}

The basic setting of game theory models a finite game as a finite set of {\em players}, each of whom chooses from a finite set of allowed {\em actions} (or {\em pure strategies}). Such a game maps any choice of actions by the players to {\em payoffs} for the players. It follows that if players are allowed to randomize over their actions, there is a well-defined notion of expected payoff for each player. Nash's famous theorem \cite{N51} states that there exist randomized (or ``mixed'') strategies for the players so that no player can improve their expected payoff by unilaterally deviating to play some alternative strategy. In this paper we make a standard assumption that all payoffs lie in the range $[0,1]$, and we take an interest in $\epsilon$-approximate Nash equilibria, in which no single player can improve their expected payoff by more than some small additive $\epsilon>0$ by deviating.

In a {\em $\lambda$-Lipschitz game}, a deviation by any single player can change any other player's payoffs by at most an additive $\lambda$. This property of a multiplayer game is of course analogous to the notion of Lipschitz continuity in mathematical analysis. As shown by Azrieli and Shmaya \cite{AS13}, for $\lambda\leq{\epsilon}/{\sqrt{8n\log(2mn)}}$, any $n$-player, $m$-action, $\lambda$-Lipschitz game has an $\epsilon$-approximate Nash equilibrium {\em in pure strategies}. This is shown by taking a mixed-strategy Nash equilibrium, known to exist by \cite{N51}, and showing that when pure strategies are sampled from it, there is a positive probability that it will be $\epsilon$-approximate.
As observed in \cite{AS13}, solutions based on {\em mixed} strategies are often criticized as being unrealistic in practice, and pure-strategy solutions are more plausible. However the existence proof of \cite{AS13} is non-constructive, raising the question of how hard they are to discover.

Here we study the problem of {\em computing} approximate pure equilibria of Lipschitz games. In a Lipschitz {\em polymatrix} game, the effect of any one player $A$'s action on the payoffs of any other player $B$ is both bounded and independent of what the remaining players are doing. For each of $B$'s available actions, $B$'s payoff is just the sum of a collection of contributions (each at most $\lambda$) associated with the action of each other player, plus some constant shift. For these simple games, we identify values of $\lambda$ and $\epsilon$ (as functions of $n$) for which this problem is \PPAD-complete, and thus unlikely to have a polynomial-time algorithm. Containment in \PPAD\ holds for more general subclasses of Lipschitz games; in essence, any for which one can deterministically and efficiently compute payoffs associated with mixed-strategy profiles. However, for Lipschitz games whose payoff function is computed by a general circuit (or more abstractly, a payoff oracle for pure-strategy profiles), we just have containment in a randomized analogue of \PPAD.

\subsection{Background, Related Work}

A key feature of the problem of computing a mixed Nash equilibrium (either an exact one, or an $\epsilon$-approximate one for some $\epsilon>0$) is that due to the guaranteed existence \cite{N51} of a solution that once found can be easily checked, it cannot be \NP-hard unless $\NP = \coNP$ \cite{MP91}. Beginning with \cite{DGP09,CDT09}, the complexity class \PPAD\ (defined below) emerged as the means of identifying whether a class of games is likely to be hard to solve. Until recently, \PPAD-completeness has mostly been used to identify computational obstacles to finding mixed equilibria rather than pure ones.  In the setting of Bayes-Nash equilibrium, Filos-Ratsikas et al. \cite{FRG21} obtain a \PPAD-completeness result for pure-strategy solutions of a class of first-price auctions. Closer to the present paper is a \PPAD-completeness result of Papadimitriou and Peng \cite{PP21} for a class of public goods games with divisible goods, where the continuous action-space of the players can represent probabilities in the solution of a corresponding 2-player game. Here instead we approximate continuous values (those taken at the gates of an arithmetic circuit for which we seek a fixed point) by the actions of a subset of the players, interpreting a given estimate as the fraction of players in such a subset who play some strategy.

\subsubsection*{Lipschitz Games}

Lipschitz games are a well-studied topic formally introduced by Azrieli and Shmaya in \cite{AS13} as a natural class of games that admit approximate pure Nash equilibria.  Such games arise in many common situations such as financial markets and various types of network games.

These games are assumed to have a large number of players, and introduce the (Lipschitz) restriction that every player can have at most a bounded impact on the payoffs of any other player.  They exhibit many fascinating properties, such as various metrics of fault tolerance (explored in \cite{GR14}) and ex-post stability of equilibria (studied in \cite{DK15}).

Following their introduction, further work on equilibria of Lipschitz games by Daskalakis and Papadimitriou \cite{DP15} shows that for {\em anonymous} Lipschitz games, pure approximate equilibria can be computed in polynomial time via reduction to \textsc{MaxFlow}, even for an approximation factor \emph{independent} of the number of players (i.e. $\epsilon=O(\lambda m)$ for $\lambda$-Lipschitz, $m$-action games).  When such games are restricted to $2$ actions, there is an algorithm that can find such equilibria making polynomially-many \emph{approximate} queries \cite{GT17} (where the query may be off by some additive error), and even a randomized algorithm running in polynomial time using only \emph{best-response} queries \cite{Bab13}.  Cummings \emph{et al}.~\cite{CKRW15} exploit a concentration argument to generalize \cite{DP15} and identify pure approximate equilibria in {\em aggregative games}.

In \cite{GK21}, Goldberg and Katzman study Lipschitz games from the perspective of {\em query complexity}, in which an algorithm has black-box access to the game's payoff function. They identify lower bounds on the number of payoff queries needed to find an equilibrium. This leaves open the problem of computing an equilibrium of a Lipschitz game that has a \emph{concise} representation: if a game is known to be concisely representable, then its query complexity is low \cite{GR16}.  Here we show that for concisely-representable Lipschitz games, there remains a computational obstacle. Prior to \cite{GK21}, \cite{GCW19} gave a query-based algorithm for computing $1/8$-approximate Nash equilibria in $1/n$-Lipschitz games; there may be scope to improve on the constant of $1/8$, but the present paper indicates a limit to further progress.

\subsubsection*{Polymatrix Games}

Another line of work explores the computational properties of equilibria in \emph{polymatrix} games introduced in \cite{Jan68}.  \cite{CD11} and \cite{DP11} show that, for \emph{coordination-only} polymatrix games (games in which each pair of players wants to agree) finding pure Nash equilibria is $\PLS$-complete while finding mixed Nash equilibria is contained in $\CLS$.  Furthermore, \cite{Rub18} shows that the problem of finding $\epsilon$-approximate Nash equilibria of \emph{general} polymatrix games is $\PPAD$-complete for some constant $\epsilon>0$.  To complement this lower bound, \cite{DFSS17} provide an upper bound in the form of an algorithm running in time polynomial in the number of players to find an $\epsilon$-approximate Nash equilibrium of such games for any constant $\epsilon>0.5$ \emph{independent of the number of actions}. Polymatrix games are known to be hard to solve even for sparse win-lose matrices \cite{LLD21}, and for tree polymatrix games \cite{DFS20tree}.

Here we extend the $\PPAD$-hardness results of \cite{Rub18} to games having the Lipschitz property in addition to being polymatrix and binary-action.  We show $\PPAD$ containment for the problem of finding approximate \emph{pure} Nash equilibria of this class.

\subsection{Our Contributions}

In its most general form, the problem we address is, for $\lambda\leq{\epsilon}/{\sqrt{8n\log4mn}}$:
\begin{quote}
\emph{Given as input an $n$-player, $m$-action game, find either an $\epsilon$-approximate pure Nash equilibrium, or a witness that the game is {\em not} $\lambda$-Lipschitz for the above $\lambda$ (a witness consists of two pure-strategy profiles whose payoffs show the game is not $\lambda$-Lipschitz).}
\end{quote}
We will consider $m$, $\epsilon$, and $\lambda$ to be parameters of the problem statement rather than provided in the input, however the problem is similar if these are instead passed in as input.

Notice that the two alternative kinds of solutions (an equilibrium \emph{or} a witness against the $\lambda$-Lipschitz property) avoids a promise problem: there is no hard-to-check promise that a given game is Lipschitz. The problem statement can be taken to apply in general to any class of games whose payoff functions are efficiently computable (e.g. circuit games \cite{SchoenebeckV12}). Efficient computability of payoffs is also a sufficient condition to guarantee that solutions are easy to check, allowing us to check the validity of either kind of solution.  In the special case of Lipschitz \emph{polymatrix} games of interest here, they can be presented in a straightforward format for which we can check the $\lambda$-Lipschitz property.

A \PPAD-completeness result consists of containment in \PPAD\ as well as \PPAD-hardness. In the context of {\em mixed-strategy} Nash equilibrium computation, the following property of a class of games is sufficient for containment in \PPAD:
\begin{quote}
\emph{Given a mixed-strategy profile, it is possible to deterministically compute the payoffs with additive error $O(\epsilon)$, in time polynomial in the representation size of the game, and $1/\epsilon$}.
\end{quote}
This property is shared by most well-studied classes of multiplayer games, such as graphical games or polymatrix games. It does not hold for unrestricted Lipschitz games\footnote{For Lipschitz games presented in terms of unrestricted circuits that compute payoffs from pure-strategy profiles, one cannot deterministically and efficiently compute approximate payoffs to mixed-strategy profiles unless $\PP=\RP$.  For games abstractly presented via a payoff query oracle, this is unconditionally not possible to achieve (regardless of any complexity class collapses).}, but it of course holds for Lipschitz polymatrix games. In this paper we obtain \PPAD\ containment for {\em pure-strategy} approximate equilibria. The approach of \cite{AS13} places the problem in ``randomized \PPAD''. We establish containment in \PPAD\ via a {\em  deterministic} approach to selecting a pure approximate equilibrium (Lemma \ref{lem:pure}). Our approach requires $\lambda$ to be a smaller-growing function of $\epsilon$ and $n$ than that of \cite{AS13}, however there \emph{is} overlap between the $(\epsilon,\lambda)$ r\'egime where the problem is contained in (deterministic) \PPAD\ and the one where the problem is hard for \PPAD.

\PPAD-hardness applies to the binary-action case where each player has two pure strategies $\{1,2\}$.  We reduce from the polymatrix games of \cite{Rub18}, by considering an induced {\em population game} (details in Section~\ref{sec:hard}), which has the Lipschitz property.

% As in \cite{Rub18}, we reduce from the \PPAD-complete problem of computing a fixpoint of a linear arithmetic circuit. Each gate in the circuit corresponds to a subset of players in the game, and the value taken at that gate corresponds to the fraction of players in that subset who play 2 as opposed to 1.

For any fixed number of actions $m\geq2$, $\PPAD$-completeness (Theorem~\ref{thm:consteps} and Corollary~\ref{cor:funceps}) holds for carefully-chosen functions $\epsilon(n)$ and $\lambda(n)$. In particular, $\lambda(n)$ needs to be sufficiently small relative to $\epsilon(n)$ in order for our derandomization approach to achieve containment in $\PPAD$ (described in Section~\ref{sec:containment}, Theorem~\ref{thm:contain}). At the same time $\lambda(n)$ needs to be large enough that we still have hardness for $\PPAD$ (Section~\ref{sec:hard}, Theorem~\ref{thm:hard}). Theorem~\ref{thm:consteps} and Corollary \ref{cor:funceps} identify $\epsilon(n)$, $\lambda(n)$ for which both of these requirements hold.

\section{Preliminaries}\label{sec:prelim}

We use the following basic notation throughout.
\begin{itemize}
\item Boldface letters denote vectors; the symbol $\veca$ is used to denote a pure profile, and $\vecp$ is used when the strategy profiles may be mixed.
\item $[n]$ and $[m]$ denote the sets $\{1,\ldots,n\}$ of players and $\{1,\ldots,m\}$ of actions, respectively.  Furthermore, $i\in[n]$ will always refer to a player, and $j\in[m]$ will always refer to an action.
\end{itemize}

\subsection{Basic Game Theoretic and Complexity Concepts}\label{subsec:game}

We first introduce standard concepts of strategy profiles, payoffs, regret, and pure/mixed equilibria, followed by the complexity class $\PPAD$.  We mostly consider binary-action games (where $m=2$) with actions labelled $1$ and $2$ (for more general results considering $m>2$, see Section \ref{sec:m}).

\paragraph*{Types of Strategy Profile}
\begin{itemize}
        \item A \emph{pure} action profile $\veca=(a_1,\ldots,a_n)\in[m]^n$ is an assignment of one action to each player. We use $\veca_{-i}=(a_1,\ldots,a_{i-1},a_{i+1},\ldots,a_n)\in[m]^{n-1}$ to denote the set of actions played by players in $[n]\setminus\{i\}$.
        \item A (possibly \emph{mixed}) strategy profile $\vecp=(p_1,\ldots,p_n)\in(\Delta([m]))^n$, where $\Delta([m])$ is the probability simplex over $[m]$. When $m=2$, let $p_i$ be the probability with which player $i$ plays action $2$. Generally $p_{i,j}$ is the probability that player $i$ allocates to action $j$. The set of distributions for players in $[n]\setminus\{i\}$ is denoted
        $\vecp_{-i}=(p_1,\ldots,p_{i-1},p_{i+1},\ldots,p_n)$.
        When $\vecp$ contains just $0$-$1$ values, $\vecp$ is equivalent to some action profile $\veca\in[m]^n$.
        \item A random realization of a mixed strategy profile $\vecp$ is a pure strategy profile $\veca$ such that, for every player $i$, $a_i=j$ with probability $p_{i,j}$.
\end{itemize}

\paragraph*{Notation for Payoffs}
    Given player $i$, action $j$, and pure action profile $\veca$,
    \begin{itemize}
        \item $u_i(j,\veca_{-i})$ is the payoff that player $i$ obtains for playing action $j$ when all other players play the actions given in $\veca_{-i}$.
        \item $u_i(\veca)=u_i(a_i,\veca_{-i})$ is the payoff that player $i$ obtains when all players play the actions given in $\veca$.
        \item Similarly for mixed-strategy profiles: $u_i(j,\vecp_{-i})=\E_{\veca_{-i}\sim\vecp_{-i}}[u_i(j,\veca_{-i})]$ and $u_i(\vecp)=\E_{\veca\sim\vecp}[u_i(\veca)]$.
    \end{itemize}
    
\paragraph*{Solution Concepts}

\begin{definition}[Best Response]\label{def:br}
    Given a player $i$ and a strategy profile $\vecp$, define the \emph{best response}
    \[\br_i(\vecp)=\argmax_{j\in[m]}u_i(j,\vecp_{-i}).\]
    In addition, for $\epsilon>0$, any action $j'$ satisfying
    \[u_i(\vecp)-u_i(j',\vecp_{-i})\leq\epsilon\]
    is labelled an \emph{$\epsilon$-best response}.
\end{definition}

\begin{definition}[Regret]\label{def:reg}
    Given a player $i$ and a strategy profile $\vecp$, define the regret
    \[\reg_i(\vecp)=\max_{j\in[m]}u_i(j,\vecp_{-i})-u_i(\vecp)\]
    to be the difference between the payoffs of player $i$'s best response to $\vecp_{-i}$ and $i$'s strategy $p_i$.
\end{definition}

\begin{definition}[Equilibria]
    We consider the following three types of Nash equilibrium:
    \begin{itemize}
        \item An {\em $\epsilon$-approximate Nash equilibrium ($\epsilon$-ANE)} is a strategy profile $\vecp^*$ such that, for every player $i\in[n]$, $\reg_i(\vecp^*)\leq\epsilon$.
        \item An {\em $\epsilon$-well supported Nash equilibrium ($\epsilon$-WSNE)} is an $\epsilon$-ANE $\vecp^*$ for which every action $j$ in the support of $p^*_i$ is an $\epsilon$-best response to $\vecp^*_{-i}$.
        \item An {\em $\epsilon$-approximate pure Nash equilibrium ($\epsilon$-PNE)} is a pure action profile $\veca$ such that, for every player $i\in[n]$, $\reg_i(\veca)\leq\epsilon$.
    \end{itemize}
\end{definition}

\paragraph*{The Complexity Class $\PPAD$}

We assume the reader is familiar with the classes of problems $\Ptime$ and $\NP$ and the concepts of asymptotic behavior (for an introduction to these concepts, see e.g. \cite{AB09}). Below we provide the detailed definition of \PPAD, introduced in \cite{Pap94}.

$\PPAD$ is defined by a problem that is complete for the class, the $\EotL$ problem.

\begin{definition}[\EotL]
Define $\EotL$ to be the problem taking as input a graph $G=(V,E)$, where the vertex set $V$ is the exponentially-large set $\{0,1\}^n$, and $E$ is encoded by two circuits $P$ and $S$, where $G$ contains directed edge $(u,v)$ if and only if $S(u)=v$ and $P(v)=u$.  Every vertex has both in-degree and out-degree at most one, and $0^n\in V$ has indegree $0$, outdegree $1$. The problem is to output another vertex $v$ with total degree $1$.
\end{definition}

\begin{definition}[$\PPAD$]\label{def:PPAD}
$\PPAD$ is the set of problems many-one reducible to $\EotL$ in polynomial time. A problem belonging to \PPAD is \PPAD-complete if $\EotL$ reduces to that problem.
\end{definition}

Here we do not use $\EotL$ directly; we reduce from another pre-existing \PPAD-complete problem.

\subsection{Lipschitz Polymatrix Games}

In this section we define the classes of games we consider, followed by the associated problem we solve.  We will first define binary-action Lipschitz games, before extending the definition to more actions.

\begin{definition}[Lipschitz Games \cite{AS13}]
For any value $\lambda\in(0,1]$, a binary-action {\em $\lambda$-Lipschitz game} is a game in which a change in strategy of any given player can affect the payoffs of any other player by at most an additive $\lambda$, i.e. for every player $i$ and pair of action profiles $\veca,\veca'$ with $a_i=a'_i$,
\[|u_i(\veca)-u_i(\veca')|\leq\lambda||\veca_{-i}-\veca'_{-i}||_0\]
(equivalently, the $\ell_1$ norm can be used, and is more relevant when it comes to mixed strategies).  Note that, under the definition of payoffs to mixed strategies presented above, because mixed strategies of binary-action games are still represented as vectors (of scalars), this can be easily extended to mixed profiles $\vecp,\vecp'$ with $p_i=p'_i$:
\[|u_i(\vecp)-u_i(\vecp')|\leq\lambda||\vecp_{-i}-\vecp'_{-i}||_1.\]
\end{definition}

The above definition, while presented in a different way, is mathematically equivalent to that given by \cite{AS13} (just provided here in a form more relevant to the current work).

In order to extend the definition above to $m>2$ actions, we will require a measure of distance for probability distributions:

\begin{definition}[Total Variation Distance]
    The \emph{total variation distance} between distributions $\mathcal{D}_1$ and $\mathcal{D}_2$ over discrete sample space $\mathcal{X}$ (each represented by $m$-dimensional vectors) is
    \[d_{\textnormal{TV}}(\mathcal{D}_1,\mathcal{D}_2)=\frac{1}{2}||\mathcal{D}_1-\mathcal{D}_2||_1=\frac{1}{2}\sum_{x\in\mathcal{X}}\left|\Pr_{x'\sim\mathcal{D}_1}(x'=x)-\Pr_{x'\sim\mathcal{D}_2}(x'=x)\right|.\]
\end{definition}

\begin{definition}\label{def:mixedLip}
For $m>2$, the pure action definition remains the same as in the $m=2$ case.  However, now that mixed strategies are represented by vectors of distributions, the Lipschitz property becomes
\[|u_i(\vecp)-u_i(\vecp')|\leq\lambda\sum_{i'=1}^nd_{\textnormal{tv}}(p_{i'},p'_{i'})\]
where $d_{\textnormal{tv}}$ is the total variation distance between two distributions (i.e. an extension of the $\ell_1$ norm definition above).
\end{definition}

\begin{remark}\label{rem:reg}
    We will rely on the total variation distance implicitly throughout the remainder of this work.  It follows from Definition \ref{def:mixedLip} that when a player $i$ moves probability $\rho$ from one action to another in a mixed strategy profile, every other player's payoffs may be affected by at most an additive $\lambda\rho$.  In addition, because the payoff function $u_i$ is $\lambda$-Lipschitz in this sense, the regret function $\reg_i$ is $2\lambda$-Lipschitz in the same way, as the payoff of player $i$'s best response can increase with a slope of at most $\lambda$ while the payoff of player $i$'s current strategy can decrease with a slope of at most $\lambda$.
\end{remark}

Here, we define \textit{polymatrix} games, introduced in \cite{Jan68}:

\begin{definition}[Polymatrix Games]
    A polymatrix game is an $n$-player, $m$-action game in which each pair of players $i_1,i_2$ plays a bimatrix game and receives as payoff the sum of their $n-1$ payoffs from each separate bimatrix game.  If player $i_1$ plays action $j_1$ and player $i_2$ plays action $j_2$, we denote the payoff to player $i_1$ from this bimatrix game as $\beta_{i_1,i_2,j_1,j_2}$.
\end{definition}

Note that, as defined, such a game is $\lambda$-Lipschitz for any
\[\lambda\geq\max_{i,i'\in[n],j,j'_1,j'_2\in[m]}\left|\beta_{i,i',j,j'_1}-\beta_{i,i',j,j'_2}\right|.\]

At most $n^2m^3$ queries to the payoff function are required to ensure this inequality holds.  Furthermore, if $i$ plays $j$, $i$'s payoff is a linear function of the indicator variables of the other player-strategy pairs. $\beta_{i,i',j,j'}$ represents a small contribution, possibly negative, that $i'$ makes to $i$ (when $i$ plays $j$) by playing $j'$.  To be valid, all total payoffs must lie in $[0,1]$, which can be checked from a game presented via the quantities $\beta_{i,i',j,j'}$ by performing the following additional checks for every $i\in[n],j\in[m]$:
\[\sum_{i'\neq i}\max_{j'\in[m]}\beta_{i,i',j,j'}\leq1\]
\[\sum_{i'\neq i}\min_{j'\in[m]}\beta_{i,i',j,j'}\geq0\]
There are $2nm$ such calculations, each of which requires $O(nm)$ operations, so an invalid game can be uncovered in polynomial time.

\begin{remark}
Lipschitz polymatrix games are a strict subset of Lipschitz games and require only $O(n^2m^2)$ queries to learn completely and can thus be concisely represented in $O(n^2m^2)$ space (by providing the value of each $\beta_{i_1,i_2,j_1,j_2}$).
\end{remark}

We now define the main problem we consider in this work.

\begin{definition}
Define $\melPLPG$ to be the problem of finding $\epsilon$-PNEs of $n$-player, $m$-action, $\lambda$-Lipschitz polymatrix games, or alternatively finding a witness that the game is not $\lambda$-Lipschitz (note once more that $m$, $\epsilon$, and $\lambda$ are parameters of the problem while $n$ is provided as input).  Furthermore, define $\melLPG$ to be the equivalent problem for mixed equilibria.
\end{definition}

\begin{remark}
Note that, while we will often fix the number of actions $m$, both $\epsilon$ and $\lambda$ are often functions of the number of players $n$. We generally think of these as being decreasing functions of $n$.
\end{remark}

We are interested in the complexity of the problem, for various pairs of functions $\epsilon(n)$, $\lambda(n)$.

\begin{remark}\label{rem:scale}
    One basic observation we make is that, for any $a\in(0,1)$, $\melPLPG$ reduces to $\PLPG{m}{a\epsilon}{a\lambda}$, by rescaling payoffs.
\end{remark}

We extend the following result of \cite{Rub18} to pure equilibria in \emph{Lipschitz} polymatrix games:

\begin{theorem}[\cite{Rub18}]\label{thm:Rub18}
    There exists some constant $\epsilon>0$ such that given a binary-action polymatrix game, finding an $\epsilon$-ANE is $\PPAD$-complete.
\end{theorem}

\section{Results}

The result we achieve is as follows:

\begin{theorem}[Main Result]\label{thm:consteps}
    For every constant $m\geq2$, there exists some constant $\epsilon>0$ such that, for all functions $\lambda(n)=\Theta(n^{-3/4})$, $\melPLPG$ is $\PPAD$-complete.
\end{theorem}

In fact, for any constant $\alpha\in(\frac{2}{3},1)$, Theorem \ref{thm:consteps} holds for $\lambda(n)=\Theta(n^{-\alpha})$.  In order to prove Theorem \ref{thm:consteps}, we will need to show both containment in $\PPAD$ and $\PPAD$-hardness under these settings of the parameters.

\subsection{Containment in \PPAD}\label{sec:containment}

\begin{theorem}\label{thm:contain}
    For all functions $\epsilon(n),\lambda(n)$ satisfying
    \[\lambda(n)=\frac{1}{\poly(n)},\qquad\epsilon=\lambda(n)\omega(n^{2/3})\]
    $\PLPG{m}{\epsilon}{\lambda}\in\PPAD$.
\end{theorem}

In particular, we prove this when $\epsilon(n)\geq6\lambda(n)\sqrt[3]{n^2m\log3m}$.  We will include the proof for $m=2$ (see Section \ref{sec:m} for the more general proof).  The proof of this theorem is broken down further into two smaller steps, along the lines of \cite{AS13}.  First, we will exhibit sufficient settings of the parameters such that \emph{mixed} approximate Nash equilibria of Lipschitz polymatrix games can be found in $\PPAD$.  Second, with the parameters set as in the statement of Theorem \ref{thm:contain}, we describe a derandomization technique for deriving a \textit{pure} approximate equilibrium from the mixed one in polynomial time.  More formally:

\begin{lemma}\label{lem:mixed}
    Whenever $\lambda(n)=\frac{1}{\poly(n)}$, $\LPG{2}{\frac{\lambda}{8}}{\lambda}\in\PPAD$.
\end{lemma}

\begin{proof}[Sketch]
As in the standard proof placing the problem of finding approximate Nash equilibria in $\PPAD$ (e.g. Theorem 3.1 of \cite{DGP09}), since $\epsilon$ is an inverse polynomial in $n$, the key requirement is the existence of an efficient deterministic algorithm answering mixed payoff queries.  Because the payoffs to a player $i$ in polymatrix games are linear in the actions of every other player, any mixed payoff query can be answered in time $O(n^2m^2)$ (or simply $O(n^2)$ when keeping the number of actions fixed).  Thus $\LPG{2}{\frac{\lambda}{8}}{\lambda(n)}\in\PPAD$.
\end{proof}

\begin{lemma}[Main Technical Lemma]\label{lem:pure}
    For functions $\epsilon(n),\lambda(n)$ satisfying $\epsilon(n)\geq\lambda\sqrt[3]{70n^2}$, an $\epsilon$-PNE of an $n$-player, binary-action, $\lambda$-Lipschitz polymatrix game $G$ can be derived from an $\frac{\lambda}{8}$-ANE of $G$ in polynomial time.
\end{lemma}

\begin{proof}
Throughout this proof we implicitly use the assumption that payoffs from mixed strategy profiles can be computed in polynomial time.  We will consider the (signed) discrepancy
\[d_i(\vecp)=u_i(2,\vecp_{-i})-u_i(1,\vecp_{-i})\]
($d_i$ is a $2\lambda$-Lipschitz function via the same argument as Remark \ref{rem:reg}).

Take some $\frac{\lambda}{8}$-ANE $\vecp^*$ of $G$.  The reduction is performed in three steps:

\begin{enumerate}[(1)]
    \item Constructing from $\vecp^*$ a well-supported Nash equilibrium $\vecp^{(0)}$.  This step is required because we will use a player's discrepancy as a linear proxy for regret, and general approximate equilibria may allow players to allocate small probabilities to bad actions, thus decoupling the two measures (in contrast with exact equilibria, in which any mixed strategy by definition implies no discrepancy between the actions).  Well-supported equilibria do not suffer from this same decoupling.
    \item Iteratively converting $\vecp^{(0)}$ into a pure action profile $\veca$ that is ``close'' to an equilibrium.  To achieve this step, we will bound the rate at which the sum of all players' discrepancies is allowed to increase, ensuring that few players experience large regret.
    \item Correcting $\veca$ to obtain an approximate pure equilibrium $\veca^*$.  Because of the guarantees from the previous step, we can achieve this by allowing every player with high regret to change to their best response.
\end{enumerate}
    
    \proofsubparagraph{\textbf{Step 1: Mixed equilibrium to WSNE}} In this first step we construct WSNE $\vecp^{(0)}$ from $\vecp^*$.  The idea here is to ensure that all players fall into one of the following two categories:
    \begin{itemize}
        \item Players with discrepancy that is small in absolute value
        \item Players playing pure strategies
    \end{itemize}
    In order to do this, consider any player $i$ such that $|d_i(\vecp^*)|>\frac{\lambda}{2}\sqrt{n}$ (without loss of generality we can assume the discrepancy is negative, i.e. action $1$ is significantly better).  Then
    \[\reg_i(\vecp^*)=-p^*_id_i>p^*_i\frac{\lambda}{2}\sqrt{n}\]
    so, since $\vecp^*$ is an approximate equilibrium,
    \[p^*_i<\frac{1}{4\sqrt{n}}.\]
    If every such player (at most $n$) changes to playing their pure best response, then every player's discrepancy will increase by at most $2n\lambda p^*_i$ (recall that discrepancy is $2\lambda$-Lipschitz), so in this new profile, no player playing a mixed strategy can experience discrepancy greater than
    \[\frac{\lambda}{2}\sqrt{n}+2n\lambda\frac{1}{4\sqrt{n}}=\lambda\sqrt{n}\]
    while no player playing a pure strategy can experience regret greater than
    \[0+2n\lambda\frac{1}{4\sqrt{n}}=\frac{\lambda}{2}\sqrt{n}.\]
    Thus this new profile is a $\lambda\sqrt{n}$-WSNE, and for any player $i$ playing a mixed strategy, $d_i(\vecp^*)\leq\lambda\sqrt{n}$.  Call this $\lambda\sqrt{n}$-WSNE $\vecp^{(0)}$.

\proofsubparagraph{\textbf{Step 2: WSNE to pure profile}}
In this step, we iteratively define intermediate strategy profiles $\vecp^{(i)}$ and intermediate sets of players $S^{(i)}$ who have, at any point, experienced small discrepancy (defined below to ensure that players not in $S^{(i)}$ cannot have high regret in profile $\vecp^{(i)}$).  Informally, for every player $i$, $\vecp^{(i)}$ is a strategy profile in which players $1,\ldots,i$ are playing pure strategies and players $i+1,\ldots,n$ may not be.  Furthermore, the set $S^{(i)}$ is the set of players who, in at least one of the profiles $\vecp^{(0)},\ldots,\vecp^{(i)}$ have experienced discrepancy at most $\lambda\sqrt{n}$.  These will be iteratively computed to ultimately arrive at $\veca=\vecp^{(n)}$: a pure action profile that is in some sense ``close'' to an approximate pure equilibrium.

Counterintutively, this will not always involve assigning each player to play their best response.  Instead, we will instruct each player to play the pure strategy that allows us to bound the effect on the other players.  Roughly, we identify below a measure of badness $C$ that is quadratic in the probability $p$ assigned by a player to strategy 2 as opposed to strategy 1. In setting $p$ to 0 or 1, we minimize the linear term in this measure, choosing 0 or 1 depending on the sign of the derivative with respect to $p$, and the measure increases by the relatively small second-order term.

    More formally, take $\vecp^{(0)}$ as above and define
    \[S^{(0)}=\{i':|d_{i'}(\vecp^{(0)})|\leq\lambda\sqrt{n}\}\]
    (i.e. a superset of the players playing mixed strategies in $\vecp^{(0)}$) while
    \[S^{(i)}=S^{(i-1)}\cup\{i':|d_{i'}(\vecp^{(i)})|\leq\lambda\sqrt{n}\}\]
    (adding at each step the set of players playing pure strategies in $\vecp^{(0)}$ but experiencing small discrepancy for the first time at step $i$; note that, due to the Lipschitz property, no player's discrepancy can change sign without that player joining $S^{(i)}$, so any player not in $S^{(i)}$ must be playing their pure best response).  Furthermore, define the sum of squared discrepancies cost function (which considers only the players in $S^{(i)}$):
    \[C(\vecp^{(i)})=\sum_{i'\in S^{(i)}}d_{i'}^2(\vecp^{(i)}).\]
    Note that $C$ only considers players in the set $S^{(i)}$, as it is meant to be a proxy for how far we are from equilibrium, and players in $[n]\setminus S^{(i)}$ have high discrepancies but are playing their best responses and hence should not contribute to this distance to equilibrium.
    
    It is difficult to minimize $C$ itself, and regardless $C$ is already only a proxy for the true objective we must minimize to achieve an approximate pure equilibrium.  However, since $d_{i'}$ is a multi-linear function (in fact, in the case of polymatrix games, linear), we can write
    \[d_{i'}(\vecp^{(i)})=c+\ell p_i\]
    where $c$ is some constant dependent only on $\vecp^{(i)}_{-i}$ and $\ell$ is a coefficient with absolute value at most $2\lambda$ (the Lipschitz parameter of $d_{i'}$).  Squaring this, and keeping in mind that $\vecp^{(i-1)}_{-i}=\vecp^{(i)}_{-i}$, we can expand (for every $i'\in S^{(i-1)}$):
    \[d_{i'}^2(\vecp^{(i)})-d_{i'}^2(\vecp^{(i-1)})=\alpha(p^{(i)}_i-p^{(i-1)}_i)+\ell^2((p^{(i)}_i)^2-(p^{(i-1)}_i)^2)\]
    for some value $\alpha_i=2c\ell$ dependent only on $\vecp^{(i)}_{-i}$.  Taking the sum over all players $i'\in S^{(i)}$, this yields
    \[C(\vecp^{(i)})-C(\vecp^{(i-1)})=A(p^{(i)}_i-p^{(i-1)}_i)+\Lambda((p^{(i)}_i)^2-(p^{(i-1)}_i)^2)+K_i\]
    where
    \begin{itemize}
        \item $A$ may be negative or non-negative.
        \item $\Lambda\leq4\lambda^2n$ (and obviously $(p^{(i)}_i)^2-(p^{(i-1)}_i)^2\leq1$).
        \item $K_i$ represents the one-time contribution of at most $\lambda^2n$ each from any player in $S^{(i)}\setminus S^{(i-1)}$ whose discrepancy will be counted in this and all future costs.  While we can't say anything about each individual $K_i$ (since we don't know when exactly every player's discrepancy will dip to this level), we note that there are at most $n$ such contributions of $\lambda^2n$, i.e:
        \[\sum_{i\in[n]}K_i\leq\lambda^2n^2\]
        (this bound is the reason we needed to start from a WSNE).
    \end{itemize}
    
    If $p^{(i-1)}_i\in\{0,1\}$, then set $\vecp^{(i)}=\vecp^{(i-1)}$ (if a player is already playing a pure strategy, we can avoid adding to $C$ at all on their turn).  Otherwise, since we can efficiently compute $A$ and decide whether it is negative or not, we can select $p^{(i)}_i$ such that $A(p^{(i)}_i-p^{(i-1)}_i)\leq0$ (one such option for $p^{(i)}_i\in\{0,1\}$ must exist).  This bounds
    \[C(\vecp^{(i)})-C(\vecp^{(i-1)})\leq4\lambda^2n+K_i\]
    and a quick induction achieves
    \[C(\vecp^{(n)})\leq4\lambda^2n^2+\sum_{i\in[n]}K_i\leq4\lambda^2n^2+\lambda^2n^2=5\lambda^2n^2.\]
    Furthermore, $\veca=\vecp^{(n)}$ is a pure action profile.
    
    \proofsubparagraph{\textbf{Step 3: Pure profile to pure equilibrium}}
    Here we distinguish between players in $S^{(n)}$ and players not in $S^{(n)}$.
    
    The players not in $S^{(n)}$ must be playing their best response in $\veca$.  This is because they were playing their best response in $\vecp^{(0)}$, and their discrepancy never approached $0$ or changed sign, so their best response must never have changed.  In particular, therefore, none of these players suffer any regret.
    
    We now consider the players in $S^{(n)}$.  Note that we have bounded the sum of the squares of the discrepancies of the players in this set to be at most $5\lambda^2n^2$.  We can now also bound the number of players with regrets that are too high and fix them.  To begin this final step, define $\delta=\lambda\sqrt[3]{20n^2}$ and consider the number of players $i\in S^{(n)}$ such that $d_i(\veca)>\delta$.  This will be at most $\frac{C(\veca)}{\delta^2}$.  So have every player with \emph{regret} at least $\delta$ in profile $\veca$ (a subset of those with high \emph{discrepancy}) simultaneously switch actions.  Then, once more invoking Remark \ref{rem:reg}, the maximum regret any player (in $S^{(n)}$ or otherwise) can experience is
    \[\delta+2\lambda \frac{C(\veca)}{\delta^2}=\delta+\frac{10\lambda^3n^2}{\delta^2}<\lambda\sqrt[3]{70n^2}.\]
    In other words, the action profile $\veca^*$ obtained after these players switch actions is a $\lambda\sqrt[3]{70n^2}$-PNE.  This completes the proof of Lemma \ref{lem:pure}.
\end{proof}

Finally, because a violation of the Lipschitz property can be uncovered in a quick check before commencing the above algorithm simply by ensuring all the coefficients are smaller than $\lambda$, and \cite{BJ12} show that $\PPAD$ is closed under Turing reductions, Theorem \ref{thm:contain} clearly follows from combining the two Lemmas above.

\subsection{Hardness for \PPAD}\label{sec:hard}

\subsubsection*{The Induced Population Game}

This section introduces the approach that will be used to show $\PPAD$-hardness.  The technique involves artificially converting a general game into a Lipschitz game by treating every player $i$ as instead a collection of many different players, each with a say in $i$'s ultimate strategy.  This reduction was used by \cite{AS13} in an alternative proof of Nash's Theorem, and by \cite{Bab13a} to upper bound the support size of $\epsilon$-ANEs.  More recently, \cite{GK21} used a query-efficient version of this reduction to lower-bound the query complexity of computing $\epsilon$-PNEs of general Lipschitz games.

\begin{definition}\label{def:gG}
    Given a game $G$ with payoff function $\vecu$, we define the {\em population game} induced by $G$, $G'=g_G(L)$ with payoff function $\vecu'$ in which every player $i$ is replaced by a population of $L$ players ($v^i_\ell$ for $\ell\in[L]$), each playing $G$ against the aggregate behavior of the other $n-1$ populations. More precisely, for $\vecp'$ a mixed profile of $G'$,
    \[u'_{v^i_\ell}(\vecp')=u_i\left(p'_{v^i_\ell},\vecp_{-i}\right)\]
    where
    \[p_{i'}=\frac{1}{L}\sum_{\ell=1}^Lp'_{v^{i'}_{\ell}}\]
    for all $i'\neq i$.
\end{definition}

\begin{remark}
    Note that, regardless of the Lipschitz parameter of $G$, the induced population game $G'=g_G(L)$ is $\frac{1}{L}$-Lipschitz.
\end{remark}

Population games date back even to Nash's thesis \cite{N50}, in which he uses them to justify the consideration of mixed equilibria as a solution concept.  To date, the reduction to the induced population game has been focused on proofs of existence.  We show that the reduction can also be used to obtain the first $\PPAD$-hardness result for a \emph{pure} Nash equilibrium problem for a class of non-Bayesian games with discrete action spaces.

\begin{remark}
    Note that any $\epsilon$-PNE of $g_G(L)$ induces a $\frac{1}{L}$-uniform\footnote{a $k$-uniform mixed strategy is one in which each action of each player is assigned a probability that is a discrete multiple of $k$} $\epsilon$-WSNE (and thus $\epsilon$-ANE) of $G$ in which each player in $G$ plays the aggregate behavior of their population.
\end{remark}

\subsubsection*{Proof of Hardness}

In Theorem \ref{thm:hard}, we only consider binary-action games (clearly the hardness results can be extended to games with more actions).

\begin{theorem}\label{thm:hard}
    There exists some constant $\epsilon>0$ such that $\PLPG{2}{\epsilon}{\epsilon n^{-4/5}}$ is $\PPAD$-hard.
\end{theorem}

\begin{remark}
    While we show this Theorem for $\lambda=\epsilon n^{-4/5}$, the same proof holds for $\lambda=\epsilon n^{-\alpha}$ for any $\alpha\in(\frac{1}{2},1)$.  Note that if $\alpha\leq\frac{1}{2}$ the pure equilibrium guarantee of \cite{AS13} does not hold.
\end{remark}

\begin{proof}[Proof of Theorem \ref{thm:hard}]
    By Theorem \ref{thm:Rub18}, there is a constant $\epsilon>0$ such that it is $\PPAD$-hard to find $\epsilon$-ANEs of polymatrix games (we may assume $\frac{1}{\epsilon}\in\N$).  We reduce the problem of computing an $\epsilon$-ANE of polymatrix game $G$ to finding an $\epsilon$-PNE of Lipschitz polymatrix game $G'$ as follows.  Consider the induced population game $G'=g_G(L)$ for $L=n^4/\epsilon^5$.  This game has $N=\left(n/\epsilon\right)^5$ players and is $\epsilon N^{-4/5}$-Lipschitz.  Thus (for large enough $n$) it is guaranteed to have an $\epsilon$-PNE.  Furthermore, $G'$ is still a polymatrix game, as the payoff to any player $v_\ell^i$ is simply the sum of the payoffs from $N-L$ games played with the players $v_\ell^{i'}$ for $i'\neq i$.
    
    Now, because the payoffs of any pure action profile of $G'$ can be derived from payoffs of mixed strategy profiles of $G$ (which can be computed in polynomial time), this entire reduction occurs in polynomial time.  Thus there is a constant $\epsilon>0$ such that, for $\lambda(n)=\epsilon n^{-4/5}$, $\PLPG{2}{\epsilon}{\lambda}$ is $\PPAD$-hard.
\end{proof}

Combining Theorems \ref{thm:contain} and \ref{thm:hard} yields Theorem \ref{thm:consteps}.  In fact, by scaling the payoffs in Theorem \ref{thm:consteps} as described in Remark \ref{rem:scale}, we obtain:

\begin{corollary}\label{cor:funceps}
    Fix some constant $\alpha\in(\frac{2}{3},1)$.  For any constant $m\geq2$, non-increasing function $\epsilon(n)=\frac{1}{\poly(n)}$, and $\lambda(n)=\Theta(\epsilon n^{-\alpha})$, $\melPLPG$ is $\PPAD$-complete.
\end{corollary}

\section{Upper Bound for Additional Actions}\label{sec:m}

\begin{theorem}[Restatement of Theorem \ref{thm:contain}]
    For functions $\epsilon(n),\lambda(n)$ satisfying
    \[\lambda(n)=\frac{1}{\poly(n)},\qquad\epsilon(n)\geq6\lambda(n)\sqrt[3]{n^2m\log3m}\]
    $\PLPG{m}{\epsilon}{\lambda}\in\PPAD$.
\end{theorem}

Once again, the proof of this theorem is broken down further into two smaller steps:

\begin{lemma}[Extension of Lemma \ref{lem:mixed}]\label{lem:mppad}
    For any fixed $m\geq2,\lambda(n)=\frac{1}{\poly(n)}$,\newline $\LPG{m}{\left(\frac{m-1}{m}\right)^2\lambda}{\lambda}\in\PPAD$.
\end{lemma}

The proof of Lemma \ref{lem:mppad} is essentially the same as that of Lemma \ref{lem:mixed}.

\begin{lemma}[Extension of Lemma \ref{lem:pure}]\label{lem:m-action}
    For functions $\epsilon(n),\lambda(n)$ satisfying
    \[\epsilon(n)\geq6\lambda(n)\sqrt[3]{n^2m\log3m}\]
    an $\epsilon(n)$-PNE of an $n$-player, $m$-action, $\lambda$-Lipschitz polymatrix game $G$ can be derived from an $\left(\frac{m-1}{m}\right)^2\lambda$-ANE of $G$ in polynomial time.
\end{lemma}

\newpage

\begin{proof}
    We will use the following values in this proof.  Let
    \[\epsilon_0=\left(\frac{m-1}{m}\right)^2\lambda\]
    \[\epsilon_1=2\sqrt{2n\lambda\epsilon_0}.\]
    Then, in order to prove Lemma \ref{lem:m-action}, we need to convert an $\epsilon_0$-ANE to an $\epsilon$-PNE.  We will do so in the same three steps as in the proof of Lemma \ref{lem:pure}, but more care must be taken when generalizing the proof.
    
    The primary generalization required is the transition from considering discrepancy to considering variance.  Whereas in Lemma \ref{lem:pure} there was a well-defined notion of discrepancy that, when squared, could be expressed as quadratic in the probabilities assigned to the actions of a given player, such a concept does not easily translate to games with more actions.  In order to rely on similar techniques, we will instead consider the variance of a subset of each player's actions.  The exact subset to be considered will be described below, and throughout this proof will be referred to as the ``relevant'' set, as it will contain the set of actions that may become relevant in our analysis of the algorithm.
    
    We will proceed along the same lines as the proof of Lemma \ref{lem:pure} with the modified approach.
    
    \proofsubparagraph{\textbf{Step 1: Mixed equilibrium to WSNE}} We begin with an $\epsilon_0$-ANE $\vecp^*$.  Consider any player $i$ and assume that player $i$ puts probability $p$ on actions with regret greater than
    \[\delta_0=\sqrt{2(n-1)\lambda\epsilon_0}.\]
    The regret of this player is at least $p\delta_0$, so it must be the case that $p\leq\epsilon_0/\delta_0$.
    
    This being the case, if every player simultaneously moves the probability from such actions to their best response, then any given player's worst-case regret becomes
    \[\delta_0+2(n-1)\lambda p\leq\delta_0+\frac{2(n-1)\lambda\epsilon_0}{\delta_0}\leq2\sqrt{2n\lambda\epsilon_0}=\epsilon_1.\]
    Thus an $\epsilon_0$-ANE can be converted easily to an $\epsilon_1$-WSNE.
    \proofsubparagraph{\textbf{Step 2: WSNE to pure profile}} This is, once again, the most involved step.  Unlike Step 1 (which occurred simultaneously) and Step 3 (which will do the same), this step will take place for each player in turn.  We will iteratively define intermediate strategy profiles $\vecp^{(i')}$ in which players $1,\ldots,i'$ are playing pure strategies and players $i'+1,\ldots,n$ may not be.  Ultimately, we will arrive at $\veca=\vecp^{(n)}$: a pure action profile that is in some sense ``close'' to an approximate pure equilibrium.
    
    Counterintuitively, this will not always involve instructing each player to play their best response when their turn arises.  Instead, we will instruct each player to play the pure strategy that allows us to bound its effect on the other players.  Roughly, we identify a measure of badness that is quadratic in the probability assigned by a player to each action.  In selecting a pure action, we minimize the linear term in this measure depending on the direction of the gradient, and the measure increases by the relatively small second-order term.  The quadratic measure we select will be the sum of the variances of subsets of the players' actions.
    
    More specifically, we define the following concepts:
    
    {\ifisarxiv\centering\fi\begin{tabular}{c|c}
        Notation & Definition \\\hline
        $S_i^{(i')}$ & \makecell{Player $i$'s \textbf{relevant set} of actions\\ (defined below) after step $i'$}\\\hline
        $m_i^{(i')}$ & \makecell{$\left|S_i^{(i')}\right|$, i.e. the number of relevant\\ actions player $i$ has after step $i'$}\\\hline
        $\vecu_i\left(S,t\right)$ & \makecell{the length-$|S|$ restriction of player $i$'s payoff\\ vector after step $t$ to the actions in set $S$}\\\hline
        $\mu_i\left(S,t\right)$ & \makecell{the average value of $\vecu_i\left(S,t\right)$, i.e. the payoff to\\ player $i$ after step $t$ for playing the uniform mixed\\ strategy over all actions in player set $S$}\\\hline
        $\sigma^2_i\left(S,t\right)$ & the variance of the values in $\vecu_i\left(S,t\right)$\\\hline
        $\reg_i\left(j,t\right)$ & \makecell{the regret player $i$ experiences\\ for playing action $j$ after step $t$}
    \end{tabular}\vspace{5mm}\\}
    Now, we will define the relevant set via the following recurrence relation.  The base case will be:
    \[S_i^{(0)}=\{j\mid\reg_i(j,0)\leq\epsilon_1\}.\]
    The recurrence is:
    \begin{lstlisting}[backgroundcolor=\color{white}]
  $S_i^{(i'+1)}=S_i^{(i')}$
  $\mu=\mu_i(S_i^{(i'+1)},i'+1)$
  while($\max_{j\not\in S_i^{(i'+1)}}u_i(j,i'+1)\geq\mu$):
      $S_i^{(i'+1)}=S_i^{(i'+1)}\cup\{\argmax_{j\not\in S_i^{(i'+1)}}u_i(j,i'+1)\}$
      $\mu=\mu_i(S_i^{(i'+1)},i'+1)$
    \end{lstlisting}
    In other words, $S_i^{(i'+1)}$ is the superset of $S_i^{(i')}$ obtained by repeatedly adding player $i$'s highest-paying action to $S_i^{(i'+1)}$ and recalculating the value of $\mu_i(S_i^{(i'+1)},i'+1)$ until doing so would decrease its value.
    We call $S_i^{(i')}$ the relevant set because it consists of the actions that, at any point during the algorithm's run, are worth consideration in our attempt to bound the total regret experienced by all players.
    
    \textbf{Step 2 of the algorithm will thus proceed as follows.}  During step $i'$, player $i'$ will first change their strategy to the pure strategy that has the least impact on the linear term of the sum of the variances of all other players' relevant sets.  \emph{Then} every player will recalculate their relevant sets using the new payoffs considering the changed strategy of player $i'$.
    
    There are three values we need to bound here:
    \begin{enumerate}
        \item The sum of the variances of $S_i^{(0)}$, the players' initial relevant sets:
        \[\sum_{i=1}^n\sigma_i^2\left(S_i^{(0)},0\right).\]
        \item How much larger the sum of the variances of $S_i^{(i'+1)}$ can be than the sum of the variances of $S_i^{(i')}$ (between steps), adding to a total over all $n$ steps of:
        \[\sum_{i=1}^n\sum_{i'=1}^n\sigma_i^2\left(S_i^{(i')},i'\right)-\sigma_i^2\left(S_i^{(i'-1)},i'\right).\]
        \item How much the sum of the variances of $S_i^{(i')},$ can increase during step $i'$, adding to a total over all $n$ steps of:
        \[\sum_{i=1}^n\sum_{i'=1}^n\sigma_i^2\left(S_i^{(i'-1)},i'\right)-\sigma_i^2\left(S_i^{(i'-1)},i'-1\right).\]
    \end{enumerate}
    Note that the sum of these three values is
    \[\sum_{i=1}^n\sigma_i^2\left(S_i^{(n)},n\right)\]
    i.e. the sum of the variances of the relevant sets at the end of the entire Step 2.
    
    We begin with quantity (1).  The variance of any set of points within the interval $[a,b]$ is at most $(b-a)^2/4$ (see, for example, \cite{BT02}), so quantity (1) is at most
    \[n\frac{\epsilon_1^2}{4}=2n^2\lambda\epsilon_0=2\left(\frac{n\lambda(m-1)}{m}\right)^2.\]
    Next, we bound quantity (2).  Quantity (2) considers no change in the strategy profile - all payoffs considered are based on strategy $\vecp^{(i')}$.  Then, for every player, actions are added to the set $S_i^{(i')}$ one by one (each time increasing the average value of the payoffs) to form the set $S_i^{(i'+1)}$.  The critical observation is that, in addition to providing a lower bound (the current average) on the payoff of this newly added action, we can also ensure an upper bound of $\mu_i\left(S_i^{(i')},i'\right)+2\lambda$, as the action must not have been relevant in the previous step.  So specifically, consider a set of $k$ points $\{x_1,\ldots,x_k\}$ with mean $\mu$ and variance $\sigma^2$ and consider adding a point $x_{k+1}$ that is within an additive $2\lambda$ of $\mu$.  We calculate the change in variance below, keeping in mind that if we subtract $\mu$ from every point the variance remains the same:
    \begin{align*}
        &\Var(x_1,\ldots,x_{k+1})-\sigma^2\\
        &=\Var(x_1-\mu,\ldots,x_{k+1}-\mu)-\sigma^2\\
        &=\left(\frac{1}{k+1}\sum_{j=1}^{k+1}(x_j-\mu)^2\right)-\left(\frac{1}{k+1}\left(\sum_{j=1}^k(x_j-\mu)\right)^2\right)-\sigma^2\\
        &=\left(\frac{1}{k+1}\left((x_{k+1}-\mu)^2+\sum_{j=1}^{k}(x_j-\mu)^2\right)\right)\\
        &\qquad-\left(\frac{1}{k+1}\left((x_{k+1}-\mu)+\sum_{j=1}^k(x_j-\mu)\right)^2\right)-\sigma^2\\
        &=\frac{(x_{k+1}-\mu)^2}{k+1}+\frac{k}{k+1}\sigma^2-\frac{(x_{k+1}-\mu)^2}{(k+1)^2}-\sigma^2\\
        &=\frac{1}{k+1}\left(\frac{k}{k+1}(x_{k+1}-\mu)^2-\sigma^2\right)
    \end{align*}
    where we can go from step 4 to step 5 using the fact that
    \[\sum_{j=1}^k(x_j-\mu)=0.\]
    So, since $|x_{k+1}-\mu|\leq2\lambda$, the largest possible increase in variance is
    \[\frac{4k\lambda^2}{(k+1)^2}\]
    (and, in fact, adding a point will usually \emph{decrease} the variance, but there's no need to get that specific because this bounded increase will serve our purpose for constant values of $m$).
    If we add this contribution for the maximum possible $m-1$ actions (the relevant set will never start empty) and all $n$ players, we see
    \[\sum_{i=1}^n\sum_{k=1}^{m-1}\frac{4k\lambda^2}{(k+1)^2}\leq n\sum_{k=1}^{m-1}\frac{4\lambda^2}{k}=4n\lambda^2H_{m-1}<4n\lambda^2(\log(m-1)+1).\]
    Finally, we need to bound quantity (3).  Note that, given $\vecu_i\left(S_i^{(i'-1)},i'-1\right)$, the quantity $\vecu_i\left(S_i^{(i'-1)},i'\right)$ and its mean and variance are all functions of the strategy vector $\vecp_{i'}$ of player $i'$.  In particular, we can write:
    \[\vecu_i(S_i^{(i'-1)},i')-\veco\mu_i(S_i^{(i'-1)},i')=\vecc_i^{(i')}+\matL_i^{(i')}\vecp_{i'}^{(i')}\]
    for some constant $m_i^{(i'-1)}\times1$ vector $\vecc_i^{(i')}$ (which may depend on the known values\newline$\vecu_i\left(S_i^{(i'-1)},i'-1\right)$) and $m_i^{(i'-1)}\times m$ coefficient matrix $\matL_i^{(i')}$ in which each element has magnitude no more than $\frac{m-1}{m}\lambda$.  Propagating this expression:
    \begin{align}
        \sigma^2_i(S_i^{(i'-1)},i')&=\frac{1}{m_i^{(i'-1)}}\left|\left|\vecu_i(S_i^{(i'-1)},i')-\veco\mu_i(S_i^{(i'-1)},i')\right|\right|_2^2\nonumber\\
        &=\frac{1}{m_i^{(i'-1)}}\left(\left|\left|\vecc_i^{(i')}\right|\right|_2^2+\left|\left|\matL_i^{(i')}\vecp_{i'}^{(i')}\right|\right|_2^2+2\vecc_i^{(i')}\cdot\matL_i^{(i')}\vecp_{i'}^{(i')}\right).\label{eq:var}
    \end{align}
    If we specifically consider the mixed term:
    \begin{align}
        2\vecc_i^{(i')}\cdot\matL_i^{(i')}\vecp_{i'}^{(i')}&=2\vecc_i^{(i)'\intercal}\left(\matL_i^{(i')}\vecp_{i'}^{(i')}\right)\nonumber\\
        &=2\left(\vecc_i^{(i')\intercal}\matL_i^{(i')}\right)\vecp_{i'}^{(i')}\nonumber\\
        &=\vecb_i^{(i')\intercal}\vecp_{i'}^{(i')}\label{eq:linvar}
    \end{align}
    for some vector $\vecb_i^{(i')}\in\R^m$.  So now we can consider the quantity we're really after.  We focus on (for some $i'\in[n]$) the sum over all players $i\in[n]$ of the changes in $\sigma_i^2\left(S_i^{(i'-1)},\cdot\right)$ after step $i'$:
    \begin{equation}\label{eq:vardiff}
        \sum_{i=1}^n\sigma_i^2\left(S_i^{(i'-1)},i'\right)-\sigma_i^2\left(S_i^{(i'-1)},i'-1\right).
    \end{equation}
    
    We want to minimize this value as best we can, as our goal is to end up with small variance among the set of relevant actions.  While the variance is not a very good proxy for the regret (see Equation \ref{eq:varlb} below), it is a value that we will be able to successfully bound.  So, substituting Equations \ref{eq:var} and \ref{eq:linvar} into Equation \ref{eq:vardiff}, we see that the change in this value across step $i'$ is
    \begin{align}
        &\sum_{i=1}^n\frac{1}{m_i^{(i'-1)}}\bigg(\left|\left|\vecc_i^{(i')}\right|\right|_2^2+\left|\left|\matL_i^{(i')}\vecp_{i'}^{(i')}\right|\right|_2^2+\vecb_i^{(i')^\intercal}\vecp_{i'}^{(i')}-\left|\left|\vecc_i^{(i')}\right|\right|_2^2+\left|\left|\matL_i^{(i')}\vecp_{i'}^{(i'-1)}\right|\right|_2^2+\vecb_i^{(i')^\intercal}\vecp_{i'}^{(i'-1)}\bigg)\nonumber\\
        &=\sum_{i=1}^n\frac{1}{m_i^{(i'-1)}}\left(\left|\left|\matL_i^{(i')}\vecp_{i'}^{(i')}\right|\right|_2^2+\vecb_i^{(i')^\intercal}\vecp_{i'}^{(i')}-\left|\left|\matL_i^{(i')}\vecp_{i'}^{(i'-1)}\right|\right|_2^2+\vecb_i^{(i')^\intercal}\vecp_{i'}^{(i'-1)}\right)\nonumber\\
        &=\sum_{i=1}^n\frac{1}{m_i^{(i'-1)}}\left(\left|\left|\matL_i^{(i')}\vecp_{i'}^{(i')}\right|\right|_2^2-\left|\left|\matL_i^{(i')}\vecp_{i'}^{(i'-1)}\right|\right|_2^2+\vecb_i^{(i')^\intercal}\left(\vecp_{i'}^{(i')}-\vecp_{i'}^{(i'-1)}\right)\right)\nonumber\\
        &\leq\sum_{i=1}^n\frac{1}{m_i^{(i'-1)}}\left(\left(\frac{m-1}{m}\lambda\right)^2+\vecb_i^{(i')^\intercal}\left(\vecp_{i'}^{(i')}-\vecp_{i'}^{(i'-1)}\right)\right)\nonumber\\
        &=\vecb^{(i')^\intercal}\left(\vecp_{i'}^{(i')}-\vecp_{i'}^{(i'-1)}\right)+\sum_{i=1}^n\frac{1}{m_i^{(i'-1)}}\left(\frac{m-1}{m}\lambda\right)^2\label{eq:vardiffupper}
    \end{align}
    for some vector $\vecb^{(i')}\in\R^m$ (the sum of the vectors $\frac{1}{m_i^{(i'-1)}}\vecb_i^{(i')}$).  So, ultimately, our algorithm aims to minimize this upper bound by ensuring that the first term is not positive.  Recall that $\vecp_{i'}^{(i'-1)}$ is a (given) vector of non-negative real numbers with $\left|\left|\vecp_{i'}^{(i'-1)}\right|\right|_1=1$, while $\vecp_{i'}^{(i')}$ is a (desired) vector containing $m-1$ zeros and a single one.  Consider the element $b$ (at index $j$) of $\vecb^{(i')}$ with the smallest value \emph{among those indices in the relevant set} $S_{i'}^{(i')}$ (note that player $i'$'s entire support will be contained in $S_{i'}^{(i')}$).  Set $\vecb=\vece_j$.  This ensures that the first term in Equation \ref{eq:vardiffupper} is not positive.
    
    So, proceeding iteratively as described, this ensures
    \[\sum_{i=1}^n\sum_{i'=1}^n\sigma_i^2\left(S_i^{(i'-1)},i'\right)-\sigma_i^2\left(S_i^{(i'-1)},i'-1\right)\leq\sum_{i=1}^n\sum_{i'=1}^n\frac{1}{m_i^{(i'-1)}}\left(\frac{m-1}{m}\lambda\right)^2\leq\left(\frac{m-1}{m}n\lambda\right)^2\]
    
    We can now bound:
    \begin{itemize}
        \item The maximum initial sum of variances of all players at $2\left(\frac{n\lambda(m-1)}{m}\right)^2$
        \item The maximum contribution to the sum of variances of all players from adding actions to the relevant set at $4n\lambda^2(\log(m-1)+1)$
        \item The maximum contribution to the sum of variances of all players from purifying the equilibrium at $\left(\frac{m-1}{m}n\lambda\right)^2$
    \end{itemize}
    Adding these together, the maximum sum of variances of actions in the relevant sets of players in $\vecp^{(n)}$ is
    \begin{align*}
    &\left(\frac{m-1}{m}n\lambda\right)^2+4n\lambda^2(\log(m-1)+1)+2\left(\frac{n\lambda(m-1)}{m}\right)^2\\
    &=3\left(\frac{m-1}{m}n\lambda\right)^2+4n\lambda^2\left(\log(m-1)+1)\right)\\
    &< 8n^2\lambda^2\log3m
    \end{align*}
    for $m>2$.
    
    \proofsubparagraph{\textbf{Step 3: Pure profile to pure equilibrium}}
    Finally, for any value of $\delta_1$, if any player has regret at least $\delta_1$, the variance of their relevant actions must be at least
    \begin{equation}\label{eq:varlb}
        \frac{\delta_1^2}{2m}
    \end{equation}
    (there must be at least a pair of relevant actions that are a distance of $\delta_1$ apart to achieve this regret, so the lowest-variance option is realized when all remaining actions yield regret $\delta_1/2$).
    So allow every player with regret greater than $\delta_1$ to instead change to playing their best response (there will be at most
    \[\frac{16n^2\lambda^2m\log3m}{\delta_1^2}\]
    such players).  As each of these players can add at most $2\lambda$ to the regret of any other player (the payoff of the best response could increase by $\lambda$ while that of the chosen action could decrease by $\lambda$), the maximum regret any player can experience is
    \[\delta_1+2\lambda\frac{16n^2\lambda^2m\log3m}{\delta_1^2}.\]
    Optimizing for $\delta_1$, this value is at most
    \[6\lambda\sqrt[3]{n^2m\log3m}.\]
\end{proof}

\section{Discussion}

Our \PPAD-hardness result is rather strongly negative, since Lipschitz polymatrix games are quite a restricted subset of either Lipschitz or polymatrix games in general. There may be scope to broaden the $(\epsilon,\lambda)$ values for which Lipschitz games are hard to solve.  On the other hand, there are more positive results for computing approximate Nash equilibria of Lipschitz games, with further scope for progress.

In particular, Theorem \ref{thm:contain} is unable to place the problem of finding $\epsilon$-PNEs of $\lambda$-Lipschitz polymatrix games in $\PPAD$ for the total range of values guaranteed by \cite{AS13} (while they guarantee existence for $\epsilon\geq\lambda\Omega(\sqrt{n\log n})$, the above successfully computes equilibria when $\epsilon\geq\lambda\Omega\left(n^{2/3}\right)$).  It would be of future interest to determine if this is a characteristic of our choice of algorithm/analysis, or an indication of a barrier to the complete derandomization of computing these equilibria.

Furthermore, it may even be the case that Lipschitz \emph{polymatrix} games, like anonymous Lipschitz games (\cite{DP15}), guarantee pure equilibria for a smaller value of $\epsilon$.  It would be interesting to determine how close to a complete derandomization we are able to achieve.  There are obvious obstacles to finding a deterministic lower bound asymptotically stronger than that of \cite{AS13}, as such a discovery would resolve the $\PP$ vs. $\BPP$ question.

\bibliography{mybibliography}

\end{document}